\documentclass[journal,12pt,draftcls,onecolumn]{IEEEtran}

\usepackage{amsmath}
\usepackage{amssymb}
\usepackage{amsmath}
\usepackage{amssymb}

\usepackage{epsfig}
\usepackage{epsf}
\usepackage{subfigure}
\usepackage{graphicx}

\def\BibTeX{{\rm B\kern-.05em{\sc i\kern-.025em b}\kern-.08em
    T\kern-.1667em\lower.7ex\hbox{E}\kern-.125emX}}

\newtheorem{claim}{Claim}

\newtheorem{theorem}{Theorem}

\newtheorem{lemma}{Lemma}

\title{Feedback Capacity of the Gaussian Interference Channel to Within 1.7075 Bits: the Symmetric Case}

\author{Changho Suh and David Tse \\
Wireless Foundations in the Department of EECS \\
University of California at Berkeley \\
Email: \{chsuh, dtse\}@eecs.berkeley.edu}

\begin{document}

\IEEEaftertitletext{
\begin{abstract}
We characterize the symmetric capacity to within 1.7075 bits/s/Hz for the two-user Gaussian interference channel with \emph{feedback}. The result makes use of a deterministic model to provide insights into the Gaussian channel. We derive a new outer bound to show that a proposed achievable scheme can achieve the symmetric capacity to within 1.7075 bits for all channel parameters. From this result, we show that feedback provides \emph{unbounded} gain, i.e., the gain becomes arbitrarily large for certain channel parameters. It is a surprising result because feedback has been so far known to provide only power gain (\emph{bounded} gain) in the context of multiple access channels and broadcast channels.
\end{abstract}
\begin{keywords}
Feedback Capacity, The Gaussian Interference Channel, A Deterministic Model
\end{keywords}
}

\maketitle

\section{Introduction}

Shannon showed that feedback does not increase capacity in the discrete-memoryless point-to-point channel \cite{shannon:it}.
However, in the multiple access channel (MAC), Gaarder and Wolf \cite{Gaarder:it} showed that feedback could increase capacity although the channel is memoryless. Inspired by this result, Ozarow \cite{Ozarow:it} found the feedback capacity region for the two-user Gaussian MAC.
However, capacity results have been open for more-than-two-user Gaussian MACs and general MACs.
Ozarow's result implies that feedback provides only power gain (\emph{bounded} gain). The reason of bounded gain is that transmitters cooperation induced by feedback can at most boost signal power (via aligning signal directions) in the MAC. Boosting signal power provides a capacity increase of a constant number of bits.

Now a question is ``Will feedback help significantly in other channels where each receiver wants to decode \emph{only desired} messages in the presence of \emph{undesired} messages (interferences)?'' To answer this question, we focus on the simple two-user Gaussian interference channel where each receiver wants to decode the messages only from its corresponding transmitter. In this channel, we show that feedback can provide \emph{unbounded} gain
for certain channel parameters. For this, we first characterize the symmetric feedback capacity for a \emph{linear} deterministic model \cite{Salman:allterton07} well capturing key properties of the Gaussian channel. Gaining insights from this model, we develop a simple two-staged achievable scheme in the Gaussian channel. We then derive a new outer bound to show that the proposed scheme achieves the symmetric capacity to within 1.7075 bits for all channel parameters.
\begin{figure}[h]
\begin{center}
{\epsfig{figure=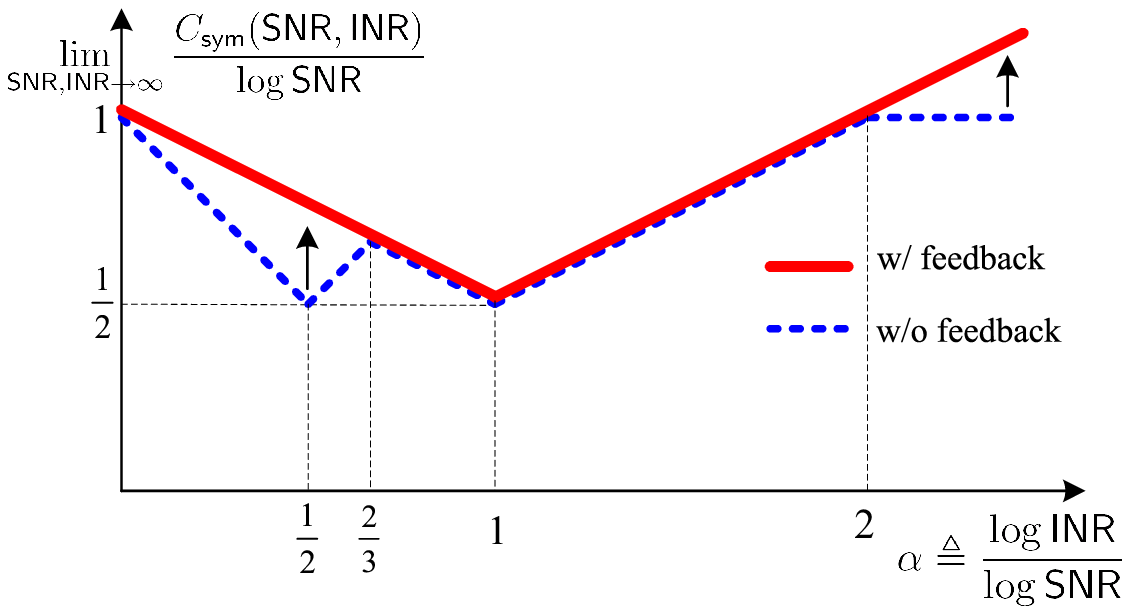, angle=0, width=0.8\textwidth}}
\end{center}
\caption{The generalized degrees-of-freedom of the Gaussian interference channel with feedback} \label{fig:gdof}
\end{figure}

The unbounded gain of feedback can be shown from the generalized degrees-of-freedom (g.d.o.f.) in Fig. \ref{fig:gdof}, defined in \cite{dtse:it07} as
\begin{align}
d(\alpha) \triangleq \lim_{\mathsf{SNR}, \mathsf{INR} \rightarrow \infty} \frac{C_{\mathsf{sym}}(\mathsf{SNR},\mathsf{INR})}{ \log \mathsf{SNR}},
\end{align}
where $\alpha$ ($x$-axis) indicates the ratio of $\mathsf{INR}$ to $\mathsf{SNR}$ in dB scale: $\alpha \triangleq \frac{\log \mathsf{INR}}{ \log \mathsf{SNR}}$.
Note that in the weak interference regime ($0 \leq \alpha \leq \frac{2}{3}$) and in the very strong interference regime ($ \alpha \geq 2$), feedback gain becomes arbitrarily large as $\mathsf{SNR}$ and $\mathsf{INR}$ go to infinity as long as channel parameters keep the certain scale so that $\alpha$ remains same. This implies \emph{unbounded} gain. This is a surprising result because feedback has been so far known to provide only power gain (\emph{bounded} gain) in the context of multiple access channels and broadcast channels \cite{Ozarow:it,Ozarow:it2}.

Some work has been done in the interference channel with feedback \cite{Kramer:it02, Kramer:it04, GastparKramer:06, Jiang:07}. In \cite{Kramer:it02, Kramer:it04}, Kramer developed a feedback strategy and derived an outer bound in the Gaussian channel; and later derived a dependence-balance outer bound with Gastpar \cite{GastparKramer:06}. However, the gap between those outer bounds and the inner bound is not tight in almost cases, except one specific set of power and channel parameters. For some channel parameters, Kramer's scheme is worse than the best known non-feedback scheme \cite{HanKoba:it81}. Recently, Jiang-Xin-Garg \cite{Jiang:07} found an achievable region in the discrete memoryless interference channel with feedback, based on the block Markov encoding \cite{Cover:it79} and the Han-Kobayashi scheme \cite{HanKoba:it81}. However, their scheme includes three auxiliary random variables requiring further optimization. Also they did not provide any upper bounds. On the other hand, we propose a simple two-staged achievable scheme which is \emph{explicit} and has \emph{only two} stages. Also we derive a tighter outer bound to do tight capacity characterization to within a constant of number of bits. Later we will provide more detailed comparison to Kramer's scheme \cite{Kramer:it02} in Section \ref{sec:Comparison}.

\section{Model}
\label{sec-DIC}

Fig. \ref{fig:GaussianDIC} (a) describes the Gaussian interference channel with feedback. We consider the symmetric interference channel where $g_{11}=g_{22}=g_d$, $g_{12}=g_{21}=g_c$, and $P_1=P_2=P$.
Without loss of generality, we assume that signal power and noise power are normalized to 1, i.e., $P_k=1$, $Z_k \sim \mathcal{CN}(0,1)$, $\forall k=1,2$. Hence, signal-to-noise ratio and interference-to-noise ratio can be defined to capture channel gains:
\begin{align}
\begin{split}
\mathsf{SNR} \triangleq |g_d|^2, \; \mathsf{INR} \triangleq |g_c|^2.
\end{split}
\end{align}

There are two independent and uniformly distributed sources, $W_k \in \left\{ 1,2,\cdots, M_k \right\}, \forall k=1,2$.
Due to feedback, the encoded signal $X_{ki}$ of user $k$ at time $i$ is a function of its own message and past output sequences:
\begin{align}
X_{ki} = f_{k}^{i} \left(W_k,Y_{k1},\cdots,Y_{k(i-1)} \right)= f_k^{i} \left(W_k,Y_{k}^{i-1} \right)
\end{align}
where we use shorthand notation $Y_{k}^{i-1}$. The symmetric capacity is defined by
\begin{align}
C_{\mathsf{sym}} = \sup \left\{R: (R,R) \in \mathcal{R} \right\},
\end{align}
where $\mathcal{R}$ is the capacity region.

\begin{figure}[h]
\begin{center}
{\epsfig{figure=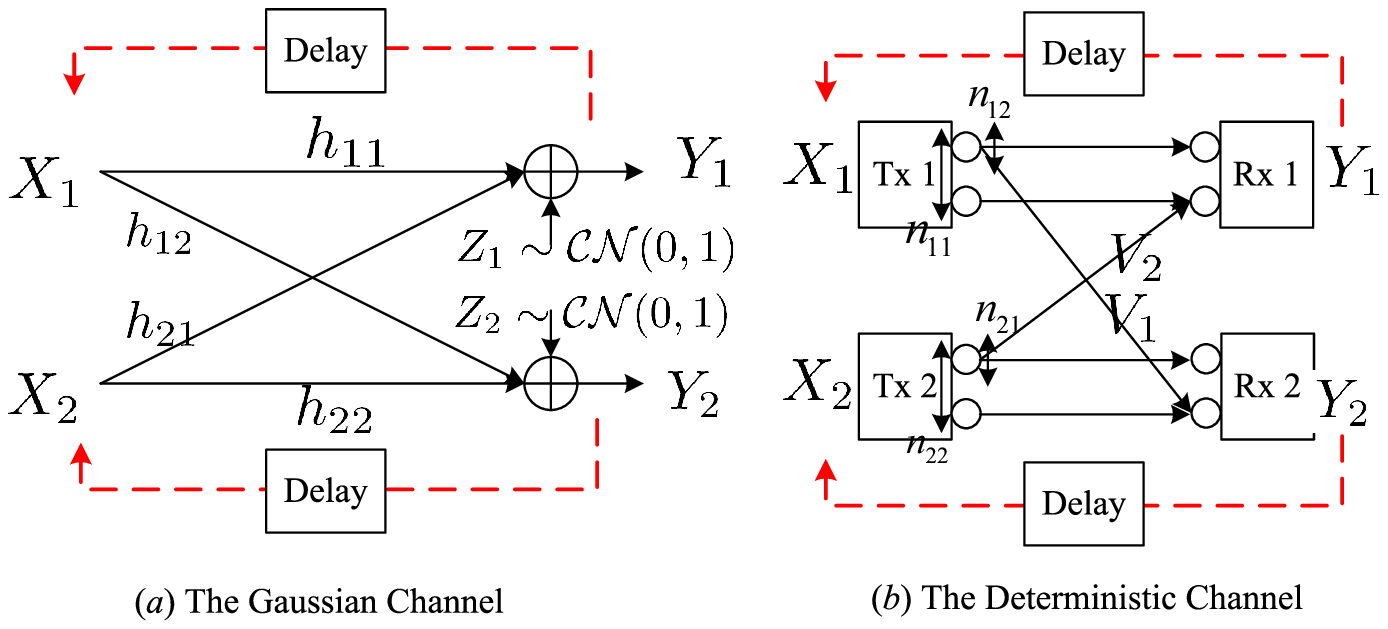, angle=0, width=0.8\textwidth}}
\end{center}
\caption{The Gaussian (and Deterministic) Interference Channels with Feedback} \label{fig:GaussianDIC}
\end{figure}


We first consider the deterministic model as shown in Fig. \ref{fig:GaussianDIC} (b). The symmetric deterministic channel is characterized by two values: $n =  n_{11} = n_{22}$ and $m =  n_{12} = n_{21}$,
where $n$ and $m$ indicate the number of signal bit levels that we can send through direct link and cross link, respectively. For each level, we assume a \emph{modulo}-2-addition. This model is useful because in the non-feedback case, the deterministic interference channel approximates the Gaussian channel within a constant gap \cite{bresler:europe}. In the feedback-case, we expect a similar constant gap as well. In the Gaussian channel, $n$ and $m$ correspond to channel gains in dB scale., i.e.,
\begin{align*}
n =  \lfloor  \log  \mathsf{SNR} \rfloor, \;m =  \lfloor  \log \mathsf{INR} \rfloor,
\end{align*}
and the modulo-2-addition corresponds to a real addition, which causes a fundamental gap between two channels. Our strategy is to first come up with a deterministic scheme, gain insights from it, and then mimic the scheme to the Gaussian channel.

\section{A Deterministic Interference Channel}
\label{sec-LinearDeterminstic}

\begin{theorem}
\label{theorem-linear-symmetric}
The symmetric feedback capacity of a deterministic interference channel is given by
\begin{align}
\begin{split}
C_{\mathsf{sym}} =  \frac{\max(n,m) + (n-m)^+}{2}.
\end{split}
\end{align}
\end{theorem}

\subsection{Proof of Achievablility}

\begin{figure}[h]
\begin{center}
{\epsfig{figure=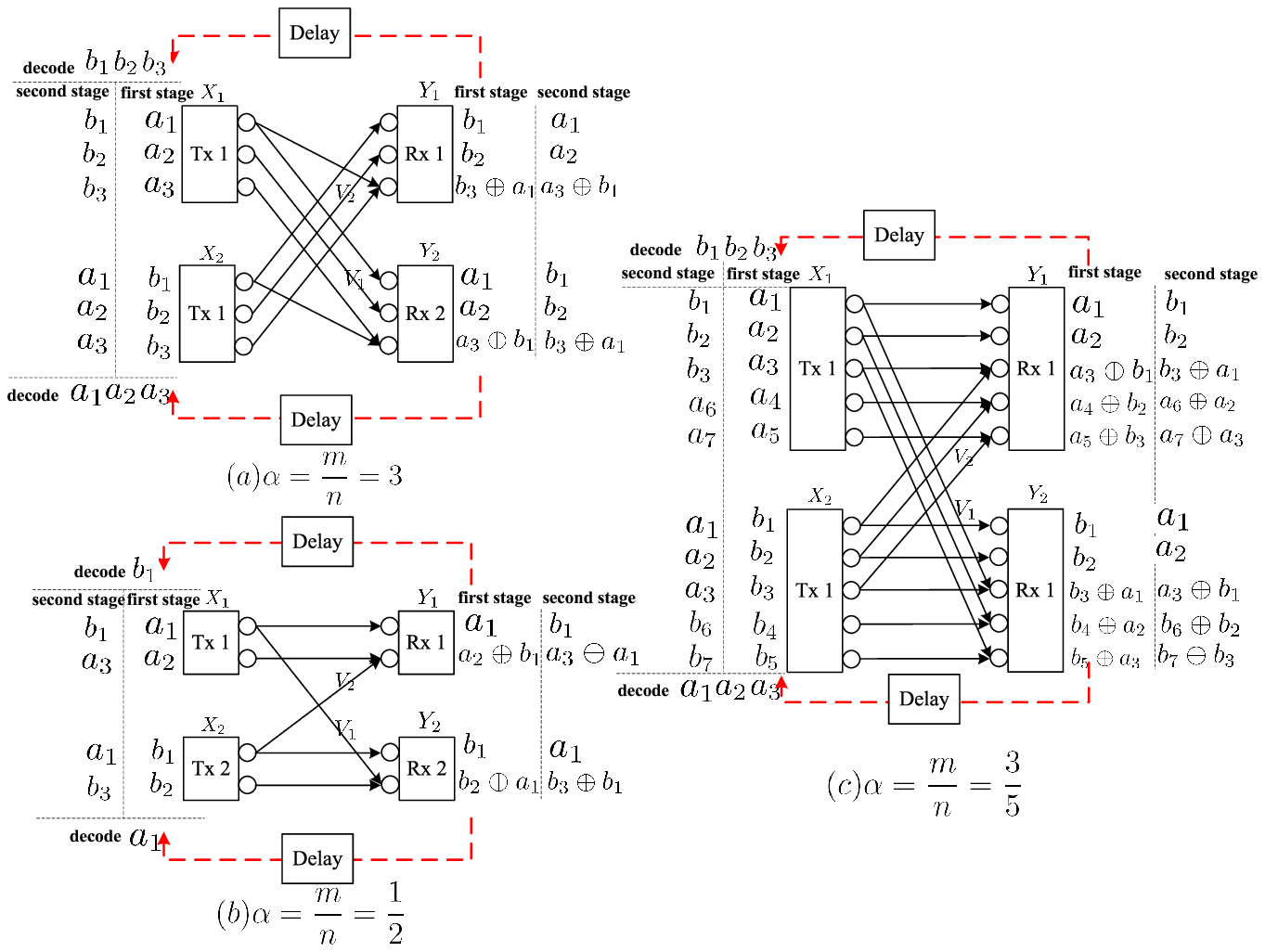, angle=0, width=0.9\textwidth}}
\end{center}
\caption{An achievable scheme of the deterministic interference channel with feedback} \label{fig:example}
\end{figure}

\textbf{Review of a non-feedback scheme \cite{bresler:europe}:}
In the non-feedback case, we typically separate into two regimes depending on the strength of interference. In the strong interference channel, the key fact is that all of the feasible rate tuples are decodable at both receivers, i.e., all messages are \emph{common}. Also since the number of received bit levels is $m$, an  achievable scheme is to send $\min(\frac{m}{2},n)$. Notice that the number of transmission bits is limited by the number $n$ of direct link bit levels. On the other hand, in the weak interference channel, only part of information is visible to the other receiver. So we spit information into two parts (the Han-Kobayashi scheme): a common part (decodable at both receivers); a private part (decodable only at the desired receiver). An achievable scheme is to send $(n-m)$ bits for private information and to send some number of bits for common information which depends on $\frac{m}{n}$. For the feedback case, we will adopt the above setting. We start with the simpler case: the strong interference regime.

\textbf{The strong interference regime ($m \geq n$):}
We will explain a scheme based on a simple example of Fig. \ref{fig:example} (a). Mimicking the non-feedback case, transmitters send only common information. The main point of a scheme is to use two stages.
In the first stage, transmitter 1 sends $a_1$, $a_2$, $a_3$ and transmitter 2 sends $b_1$, $b_2$, $b_3$. Note that each transmitter sends the whole $m$ bits instead of $\min(\frac{m}{2},n)$ (the number of bits sent in the non-feedback case). Due to this, each receiver needs to defer decoding to the second stage. In the second stage, using feedback, each transmitter decodes information of the other user, e.g., transmitter 1 decodes $b_1$, $b_2$, $b_3$ and transmitter 2 decodes $a_1$, $a_2$, $a_3$. Each transmitter then sends information of the other user.

Now each receiver can decode its own data by subtracting the received signal in the first stage from the second. Receiver 1 decodes $a_1$, $a_2$, $a_3$ by subtracting $b_1$ from the second received signal. Notice that the second stage was used for refining all bits sent previously, without sending additional information. Therefore, the symmetric rate is $\frac{3}{2}$. Considering the general case $(n,m)$, we achieve
\begin{align}
R_{\mathsf{sym}} = \frac{m}{2}.
\end{align}

Note that with feedback, the symmetric rate can exceed $n$ bits, which was the limit for the non-feedback case. 
This is because the very strong interfering link helps significantly to \emph{relay} other messages through feedback. For example, the information flow of user 1 is through \emph{indirect} links ($X_{1} \rightarrow V_{1} \rightarrow \textrm{feedback} \rightarrow X_2 \rightarrow V_2 \rightarrow Y_1$) instead of direct link ($X_{1} \rightarrow Y_{1}$). This concept coincides with \emph{correlation routing} in Kramer's paper \cite{Kramer:it02}.

\textbf{The weak interference regime ($m < n$):}
We will explain a scheme based on an example of Fig. \ref{fig:example} (b).
Similar to the non-feedback case, information is split into two parts. But it has two stages.
In the first stage, transmitter 1 sends private information $a_2$ on the lower level (invisible to the other receiver) and common information $a_1$ on the upper signal level (visible to the other receiver). Similarly transmitter 2 sends $b_1$ and $b_2$. Similar to the non-feedback case, each transmitter sends $(n-m)$ private bits. However, there is a difference in sending common information. Each transmitter sends $m$ common bits whatever $\frac{m}{n}$ is, unlike the non-feedback case where the number of common bits depends on $\frac{m}{n}$. Then, receiver 1 gets the clean signal $a_1$ on the upper level and the interfered signal $a_2 \oplus b_1$ on the lower level.
In this case ($\alpha = \frac{1}{2}$), receiver 1 can decode its common information $a_1$ in the first stage. However, for the other case, e.g., $\alpha = \frac{3}{5}$ (Fig. \ref{fig:example} (c)), receiver 1 cannot fully decode common information in the first stage because a part of it is interfered by common information of the other user. Therefore, each receiver needs to defer decoding to the second stage.

In the second stage, with feedback, each transmitter can decode common information of the other user. Transmitter 1 and 2 can decode $b_1$ and $a_1$, respectively. Each transmitter then sends common information of the other user on the upper level. Sending this, receiver 1 can refine the corrupted symbol received in the first stage without causing any interferences to the other receiver. On the lower level, each transmitter sends new private information. Transmitter 1 and 2 send $a_3$ and $b_3$, respectively.

Using the first and second received signals, receiver 1 can now decode the corrupted symbol $a_2$ sent in the first stage. At the same time, it can decode new private information $a_3$ by stripping $a_1$. During two stages, each receiver can decode three symbols out of two levels. Therefore, the symmetric rate is $\frac{3}{2 \cdot 2}$. This scheme can be easily generalized into the case of $(n, m)$. During two stages, each receiver can decode all of the messages sent in the first stage and a new private message sent in the second stage. Therefore, the symmetric rate is

\begin{align}
R_{\mathsf{sym}} = \frac{n + (n-m) }{2} = n -\frac{m}{2}.
\end{align}

\textbf{Remarks on the achievable scheme:}
Our two-staged scheme has some similarity with an achievable scheme in \cite{Jiang:07} in that using feedback each transmitter decodes common information of the other user. However, our scheme is different since it is \emph{explicit} and has \emph{only two} stages, while the scheme in \cite{Jiang:07} employs three auxiliary random variables (requiring further optimization) and the block Markov encoding (requiring a long block length).

\subsection{Proof of Converse}
We have
\begin{align*}
\begin{split}
N&(R_1 + R_2)= H(W_1) + H(W_2) \overset{(a)}{=}  H(W_1|W_2)+ H(W_2) \\
&\overset{(b)}{\leq} I(W_1;Y_1^{N}|W_2) + I(W_2;Y_2^{N}) + N \epsilon_N \\
&\overset{(c)}{=} H(Y_1^{N}|W_2) + I(W_2;Y_2^{N}) + N \epsilon_N \\
&\leq H(Y_1^{N},V_1^{N}|W_2) + I(W_2;Y_2^{N}) + N \epsilon_N \\
&=H(Y_1^{N}|V_1^{N},W_2) + H(Y_2^{N}) + \left[H(V_1^{N}|W_2) - H(Y_2^{N}|W_2)\right] +  N \epsilon_N \\
&\overset{(d)}{=} H(Y_1^{N}|V_1^{N},W_2) + H(Y_2^{N}) +  N \epsilon_N \\
&\overset{(e)}{=} H(Y_1^{N}|V_1^{N},W_2,X_2^{N},V_2^{N}) + H(Y_2^{N}) +  N \epsilon_N \\
&\overset{(f)}{\leq} \sum_{i=1}^{N} \left[ H(Y_{1i}|V_{1i},V_{2i}) + H(Y_{2i}) \right] + N \epsilon_N
\end{split}
\end{align*}
where ($a$) follows from the independence of $W_1$ and $W_2$; ($b$) follows from Fano's inequality; ($c$) follows from the fact that $Y_1^{N}$ is a function of $W_1$ and $W_2$; ($d$) follows from  $H(V_1^N|W_2)=H(Y_2^N|W_2)$ (see Claim \ref{claim-2}); ($e$) follows from the fact that $X_2^{N}$ is a function of $(W_2,V_1^{N-1})$ (see Claim \ref{claim-1}) and $V_2^{N}$ is a function of $X_2^{N}$; ($f$) follows from the fact that conditioning reduces entropy.

\begin{claim}
\label{claim-2}
$H(V_1^{N}|W_2) = H(Y_2^{N}|W_2).$
\end{claim}
\begin{proof}

\begin{align*}
\begin{split}
H&(Y_2^{N}|W_2) = \sum_{i=1}^{N}H(Y_{2i}|Y_2^{i-1},W_2) \\
&\overset{(a)}{=} \sum_{i=1}^{N}H(V_{1i}|Y_2^{i-1},W_2) \\
&\overset{(b)}{=}  \sum_{i=1}^{N}H(V_{1i}|Y_2^{i-1},W_2,X_2^{i},V_1^{i-1}) \\
&\overset{(c)}{=}  \sum_{i=1}^{N}H(V_{1i}|W_2,V_1^{i-1}) =H(V_1^{N}|W_2),
\end{split}
\end{align*}
where ($a$) follows from the fact that  $Y_{2i}$ is a function of $(X_{2i},V_{1i})$ and $X_{2i}$ is a function of $(W_2, Y_2^{i-1})$; ($b$) follows from the fact that $X_2^{i}$ is a function of $(W_2,Y_2^{i-1})$ and $V_2^{i}$ is a function of $X_2^{i}$; ($c$) follows from the fact that $Y_{2}^{i-1}$ is a function of $(X_2^{i-1},V_1^{i-1})$ and $X_2^{i}$ is a function of  $(W_2,V_1^{i-1})$ (by Claim \ref{claim-1}).

\end{proof}

\begin{claim}
\label{claim-1}
For all $i\geq 1$, $X_1^{i}$ is a function of $(W_1,V_2^{i-1})$ and $X_2^{i}$ is a function of $(W_2,V_1^{i-1})$.
\end{claim}
\begin{proof}
By symmetry, it is enough to prove only one. Since the channel is deterministic (noiseless), $X_1^{i}$ is a function of $W_1$ and $W_2$. In Fig. \ref{fig:GaussianDIC} (b), we can easily see that information of $W_2$ delivered to the first link must pass through $V_{2i}$. Also note that $X_{1i}$ depends on the past output sequences until $i-1$ (due to feedback delay). Therefore, $X_1^{i}$ is a function of $(W_1,V_2^{i-1})$.
\end{proof}

Now let the time index $Q$ be a random variable uniformly distributed over the set $\{1,2,\cdots,N\}$ and independent of $(W_1,W_2,X_1^N,X_2^N,Y_1^N,Y_2^N)$.
We define $X_k = X_{kQ}, \;V_k = V_{kQ}, Y_k = Y_{kQ}, \forall k=1,2$. If $(R_1,R_2)$ is achievable, then $\epsilon_N \rightarrow 0$ as $N \rightarrow \infty$. Hence, we get
\begin{align*}
R_1 + R_2 &\leq  H(Y_{1}|V_{1},V_{2}) + H(Y_{2}).
\end{align*}
Since the RHS is maximized when $X_1$ and $X_2$ are uniform and independent, we get
\begin{align}
C_{\mathsf{sym}} \leq \frac{\max(n,m) + (n-m)^+}{2}.
\end{align}
This establishes the converse.

\section{The Gaussian Interference Channel}
\label{sec:GIC}

\subsection{An Achievable Rate}

\begin{theorem}
\label{theorem:GaussianAchievable}
In the strong Gaussian interference channel ($\mathsf{INR} \geq \mathsf{SNR}$), we can achieve
\begin{align}
\begin{split}
\label{eq:Rsym_strong}
R_{\mathsf{sym}}^{\mathsf{strong}} =  \frac{1}{2} \log \left(1 + \mathsf{INR} \right).
\end{split}
\end{align}

In the weak Gaussian interference channel ($\mathsf{INR} \leq \mathsf{SNR}$), we can achieve
\begin{align}
\begin{split}
\label{eq:Rsym_weak}
R_{\mathsf{sym}}^{\mathsf{weak}} =
\left\{
  \begin{array}{ll}
    \log \left( 1+ \frac{\mathsf{SNR}}{2\mathsf{INR}} \right) +  \frac{1}{2} \log \left(1 + \mathsf{INR} \right)   - \frac{1}{2}, & \hbox{$\mathsf{INR} \geq 1$;} \\
   \log \left( 1 + \frac{\mathsf{SNR}}{\mathsf{INR}+1}  \right), & \hbox{$\mathsf{INR} \leq 1$.}
  \end{array}
\right.
\end{split}
\end{align}

\end{theorem}

\begin{proof}

\begin{figure}[h]
\begin{center}
{\epsfig{figure=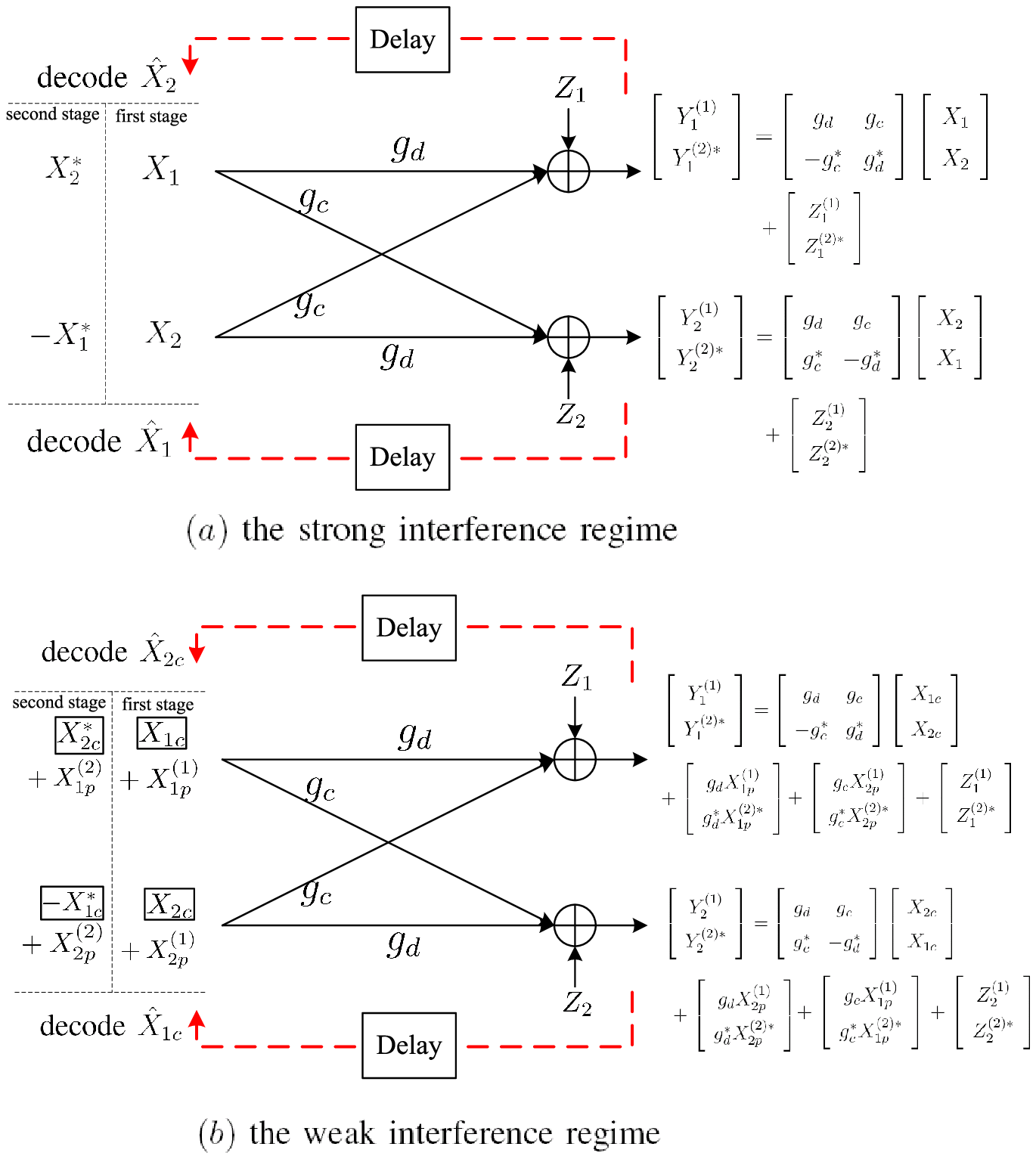, angle=0, width=0.8\textwidth}}
\end{center}
\caption{An achievable scheme of the Gaussian interference channel with feedback.
} \label{fig:Gaussian_example}
\end{figure}

\textbf{The strong interference regime ($\mathsf{INR} \geq \mathsf{SNR}$, Fig. \ref{fig:Gaussian_example} (a)):}
Mimicking the deterministic case, each transmitter sends only common information and employs two stages. In the first stage, each transmitter sends its own signal. In the second stage, each transmitter sends information of the other user after decoding it with the help of feedback. In the Gaussian \emph{noisy} channel, we need to be careful in how to combine the received signals (during two stages) to decode the message.
Alamouti's scheme \cite{Alamouti:JSAC} gives insights into this. Notice that with feedback both messages are available at transmitters in the second stage. However, in spite of knowing both messages, transmitters cannot control the messages already sent in the first stage. Hence, they can \emph{partially} collaborate in the second stage. However, the beauty of Alamouti's scheme is that messages can be designed to be \emph{orthogonal} (for two time slots), although the messages in the first time slot are sent \emph{without any coding}. This was well exploited and pointed out in \cite{Laneman:it03}. Therefore, with Alamouti's scheme, transmitters are able to encode messages so that those are \emph{orthogonal}. In the interference channel, orthogonality between different messages guarantees to completely remove the other message (interference). This helps improving performance significantly.

In the first stage (block), transmitter 1 and 2 send codewords $X_{1}^{N}$ and $X_{2}^{N}$ with rates $R_1$ and $R_2$, respectively. In the second stage, using feedback, transmitter 1 decodes $X_{2}^{N}$ by stripping its own codeword $X_{1}^{N}$.
This can be decoded if
\begin{align}
\label{eq:strongR_2c_constraint}
R_{2}  \leq \frac{1}{2} \log  \left( 1 + \mathsf{INR} \right) \;\;bits/s/Hz.
\end{align}
Similarly transmitter 2 decodes $X_{1}^{N}$; hence, we have the same constraint for $R_{1}$.

Now we apply Alamouti's scheme. In the second stage, transmitter 1 sends $X_{2}^{N*}$ and transmitter 2 sends $-X_{1}^{N*}$. Then, receiver 1 can gather the signals received during two stages. For $1 \leq i \leq N$,
\begin{align}
\left[
  \begin{array}{c}
    Y_{1i}^{(1)} \\
    Y_{1i}^{(2)*} \\
  \end{array}
\right]
= &\left[
    \begin{array}{cc}
      g_d & g_c \\
      -g_c^{*} & g_d^{*} \\
    \end{array}
  \right]
\left[
  \begin{array}{c}
    X_{1i} \\
    X_{2i} \\
  \end{array}
\right] +  \left[
  \begin{array}{c}
    Z_{1i}^{(1)} \\
    Z_{1i}^{(2)*} \\
  \end{array}
\right]
\end{align}
To decode $X_{1i}$, receiver 1 multiplies the row vector orthogonal to the vector corresponding to $X_{2i}$ so we get
\begin{align}
\left[ \begin{array}{cc}
         g_d^{*} & -g_c
       \end{array}
 \right] \left[
  \begin{array}{c}
    Y_{1i}^{(1)} \\
    Y_{1i}^{(2)*} \\
  \end{array}
\right] = (|g_d|^2 + |g_c|^2) X_{1i} + g_d^* Z_{2i}^{(1)} - g_c Z_{1i}^{(2)*}.
\end{align}
Then, the codeword $X_{1}^{N}$ can be decoded if
\begin{align}
\label{eq:strongR_1c_constraint}
R_{1} \leq \frac{1}{2} \log \left( 1 + \mathsf{SNR} +  \mathsf{INR} \right) \;\; bits/s/Hz.
\end{align}
Similar operations are done at receiver 2.
Since  (\ref{eq:strongR_1c_constraint}) is implied by (\ref{eq:strongR_2c_constraint}), we get the desired result (\ref{eq:Rsym_strong}).

\textbf{The weak interference regime ($\mathsf{INR} \leq \mathsf{SNR}$, Fig. \ref{fig:Gaussian_example} (b)):} Similar to the deterministic case, a scheme has two stages and information is split into common and private parts. Also recall that in the deterministic case, only common information is sent twice during two stages. Therefore, a natural idea is to apply Alamouti's scheme only for common information. Private information is newly sent for both stages.

In the first stage, transmitter 1 independently generates a common codeword $X_{1c}^{N}$ and a private codeword $X_{1p}^{N,(1)}$ with rates $R_{1c}$ and $R_{1p}^{(1)}$, respectively. For power splitting, we adapt the idea of the simplified Han-Kobayashi scheme \cite{dtse:it07} where private power is set such that private information is in the noise level:
The scheme is to set
\begin{align}
\begin{split}
\label{eq:weakpowersplit}
\lambda_{p} = \min \left( \frac{1}{\mathsf{INR}}, 1 \right), \;\; \lambda_{c} = 1- \lambda_{p},
\end{split}
\end{align}
Now we assign power $\lambda_{p}$ and $\lambda_{c}$ to $X_{1p,i}^{(1)}$ and $X_{1c,i}$, $\forall i$, respectively; and \emph{superpose} two signals to form  channel input. Similarly transmitter 2 sends $X_{2p}^{N,(1)} + X_{2c}^{N}$. By symmetry, we use the same $\lambda_{p}$ and $\lambda_{c}$.
Similar to the deterministic case, each receiver defers decoding to the second stage.

In the second stage, $Y_1^{N,(1)}$ is available at transmitter 1 by feedback. Transmitter 1 subtracts its own codeword $X_{1p}^{N,(1)}$ and $X_{1c}^{N}$ from $Y_1^{N,(1)}$ and then decodes a common codeword $X_{2c}^{N}$ of the other user. We can decode this if
\begin{align}
\label{eq:weakR_2c_constraint}
R_{2c} \leq \frac{1}{2} \log \left( 1 + \frac{\lambda_c \mathsf{INR} }{ \lambda_p \mathsf{INR} + 1 } \right) \;\; bits/s/Hz.
\end{align}
Similarly transmitter 2 decodes $X_{1c}^{N}$. We have the same constraint for $R_{1c}$.

Now we apply Alamouti's scheme only for common information $X_{1c}^{N}$ and $X_{2c}^{N}$.
Transmitter 1 sends $X_{2c}^{N*}$ and just adds new private information $X_{1p}^{N,(2)}$.  On the other hand, transmitter 2 sends $-X_{1c}^{N*}$ and $X_{2p}^{N,(2)}$. Then, receiver 1 gets

\begin{align}
\begin{split}
\left[
  \begin{array}{c}
    Y_{1i}^{(1)} \\
    Y_{1i}^{(2)*} \\
  \end{array}
\right]
= \left[
    \begin{array}{cc}
      g_d & g_c \\
      -g_c^* & g_d^* \\
    \end{array}
  \right]
\left[
  \begin{array}{c}
    X_{1c,i} \\
    X_{2c,i} \\
  \end{array}
\right]
+ \left[
  \begin{array}{c}
    g_d X_{1p,i}^{(1)} \\
    g_d^* X_{1p,i}^{(2)*} \\
  \end{array}
\right]  +
\left[
  \begin{array}{c}
    g_c X_{2p,i}^{(1)} \\
    g_c^* X_{2p,i}^{(2)*} \\
  \end{array}
\right]
+ \left[
  \begin{array}{c}
    Z_{1i}^{(1)} \\
    Z_{1i}^{(2)*} \\
  \end{array}
\right].
\end{split}
\end{align}

To decode $X_{1c,i}$, we consider
\begin{align}
\begin{split}
\left[ \begin{array}{cc}
         g_d^{*} & -g_c
       \end{array}
 \right] \left[
  \begin{array}{c}
    Y_{1i}^{(1)} \\
    Y_{1i}^{(2)*} \\
  \end{array}
\right] &= (|g_d|^2 + |g_c|^2) X_{1c,i} + |g_d|^2 X_{1p,i}^{(1)} - g_c g_d^* X_{1p,i}^{(2)*} \\
 & + g_d^*g_c X_{2p,i}^{(1)} - |g_c|^2 X_{2p,i}^{(2)*}  + g_d^* Z_{1i}^{(1)} - g_c Z_{1i}^{(2)*}.
\end{split}
\end{align}
$X_{1c}^{N}$ can be decoded if
\begin{align}
\label{eq:weakR_1c_constraint}
R_{1c} \leq  \frac{1}{2} \log \left( 1 + \frac{ \lambda_c \left( \mathsf{SNR} +  \mathsf{INR} \right) }{ \lambda_p \left( \mathsf{SNR} + \mathsf{INR}  \right) + 1 } \right) \;\; bits/s/Hz.
\end{align}
This constraint (\ref{eq:weakR_1c_constraint}) is implied by (\ref{eq:weakR_2c_constraint}).

Similarly receiver 1 can decode $X_{2c}^{N}$, so we have the same constraint for $R_{2c}$.
Now receiver 1 subtracts ($X_{2c}^{N}$, $X_{2c}^{N}$) and then decodes $X_{1p}^{N,(1)}$ and $X_{1p}^{N,(2)}$. This can be decoded if
\begin{align}
\label{eq:weakR_1p_constraint}
R_{1p}^{(1)} \leq \log \left( 1 + \frac{\lambda_p \mathsf{SNR} }{ \lambda_p \mathsf{INR} + 1 } \right) \;\; bits/s/Hz.
\end{align}
Similar operations are done at receiver 2.
Under the simple power setting (\ref{eq:weakpowersplit}), we get the desired result (\ref{eq:Rsym_weak}).

\end{proof}

\textbf{Remarks on the achievable scheme:}
The Alamouti-based phase rotating technique in our scheme looks similar to phase rotating techniques in \cite{Ozarow:it} and \cite{Kramer:it02}. However, it is different because phase rotating in our scheme is to make the desired signal \emph{orthogonal} to interference, while the purpose of \cite{Ozarow:it, Kramer:it02} is to \emph{align} the phase of desired signal to boost power gain. Also our scheme is essentially different from other schemes \cite{Ozarow:it, Kramer:it02} which are based on Schalkwijk-Kailath scheme.
%

\subsection{An Outer Bound}

\begin{theorem}
\label{theorem:GaussianUpperBound}
The symmetric capacity of the Gaussian interference channel with feedback is upper-bounded by

\begin{align}
\begin{split}
\label{eq:GaussianUpperBound}
C_{\mathsf{sym}} \leq \frac{1}{2} \sup_{0 \leq \rho \leq 1} \left[ \log \left(1 + \frac{(1-\rho^2) \mathsf{SNR}}{ 1 + (1-\rho^2)\mathsf{INR}} \right) + \log \left( 1 + \mathsf{SNR} + \mathsf{INR} + 2 \rho \sqrt{\mathsf{SNR} \cdot \mathsf{INR} } \right) \right].
\end{split}
\end{align}

\end{theorem}

\begin{proof}

For side information, we consider a noisy version of $V_1$:
\begin{align}
S_1 = V_1 + Z_2 = g_c X_1 + Z_2.
\end{align}
Using this, we get
\begin{align*}
\begin{split}
N&(R_1 + R_2)= H(W_1) + H(W_2) =  H(W_1|W_2)+ H(W_2) \\
&\leq I(W_1;Y_1^{N}|W_2) + I(W_2;Y_2^{N}) + N \epsilon_N \\
&\overset{(a)}{\leq} I(W_1;Y_1^{N},S_1^{N}|W_2) + I(W_2;Y_2^{N}) + N \epsilon_N \\
&= h(Y_1^{N},S_1^{N}|W_2)  - h(Y_1^{N},S_1^{N}|W_1,W_2) + I(W_2;Y_2^{N}) + N \epsilon_N \\
&\overset{(b)}{=} h(Y_1^{N},S_1^{N}|W_2) - \sum \left[ h(Z_{1i}) + h(Z_{2i}) \right]  + I(W_2;Y_2^{N}) + N \epsilon_N \\
&= h(Y_1^{N}|S_1^{N},W_2)  - \sum  h(Z_{1i})  + h(Y_2^{N}) - \sum h(Z_{2i}) \\
&\quad + \left[h(S_1^{N}|W_2) - h(Y_2^{N}|W_2)\right] +  N \epsilon_N \\
&\overset{(c)}{=} h(Y_1^{N}|S_1^{N},W_2) - \sum  h(Z_{1i}) + h(Y_2^{N})  - \sum h(Z_{2i})  +  N \epsilon_N \\
&\overset{(d)}{=} h(Y_1^{N}|S_1^{N},W_2,X_2^{N}) - \sum h(Z_{1i}) + h(Y_2^{N}) - \sum h(Z_{2i})  +  N \epsilon_N \\
&\overset{(e)}{\leq} \sum_{i=1}^{N} \left[ h(Y_{1i}|S_{1i},X_{2i}) - h(Z_{1i}) + h(Y_{2i}) - h(Z_{2i}) \right] +  N \epsilon_N
\end{split}
\end{align*}
where ($a$) follows from the fact that adding information increases mutual information; ($b$) follows from $h(Y_1^{N},S_1^{N}|W_1,W_2) =\sum \left[ h(Z_{1i}) + h(Z_{2i}) \right]$ (see Claim \ref{claim-5}); $(c)$ follows from  $h(S_1^N|W_2)=h(Y_2^N|W_2)$ (see Claim \ref{claim-4}); ($d$) follows from the fact that $X_2^{N}$ is a function of $(W_2,S_1^{N-1})$ (see Claim \ref{claim-3}); ($e$) follows from the fact that conditioning reduces entropy.

\begin{claim}
\label{claim-5}
$h(Y_1^{N},S_1^{N}|W_1,W_2) = \sum \left[ h(Z_{1i}) + h(Z_{2i}) \right].$
\end{claim}
\begin{proof}

\begin{align*}
\begin{split}
h&(Y_1^{N},S_1^{N}|W_1,W_2) = \sum h(Y_{1i},S_{1i}|W_1,W_2,Y_1^{i-1},S_1^{i-1}) \\
&\overset{(a)}{=} \sum h(Y_{1i},S_{1i}|W_1,W_2,Y_1^{i-1},S_1^{i-1},X_{1i},X_{2i}) \\
&\overset{(b)}{=}  \sum h(Z_{1i},Z_{2i}|W_1,W_2,Y_1^{i-1},S_1^{i-1},X_{1i},X_{2i}) \\
&\overset{(c)}{=}  \sum \left[ h(Z_{1i}) + h(Z_{2i}) \right],
\end{split}
\end{align*}
where ($a$) follows from the fact that $X_{1i}$ is a function of $(W_1, Y_1^{i-1})$ and $X_{2i}$ is a function of $(W_2,S_{1}^{i-1})$ (by Claim \ref{claim-3});
($b$) follows from the fact that $Y_{1i} = g_d X_{1i} + g_c X_{2i} + Z_{1i}$ and $S_{1i} = g_c X_{1i} + Z_{2i}$; ($c$) follows from the memoryless property of the channel and the independence assumption of $Z_{1i}$ and $Z_{2i}$.
\end{proof}

\begin{claim}
\label{claim-4}
$h(S_1^{N}|W_2) = h(Y_2^{N}|W_2).$
\end{claim}
\begin{proof}

\begin{align*}
\begin{split}
h&(Y_2^{N}|W_2) = \sum h(Y_{2i}|Y_2^{i-1},W_2) \\
&\overset{(a)}{=} \sum h(S_{1i}|Y_2^{i-1},W_2) \\
&\overset{(b)}{=}  \sum h(S_{1i}|Y_2^{i-1},W_2,X_2^{i},S_1^{i-1}) \\
&\overset{(c)}{=}  \sum h(S_{1i}|W_2,S_1^{i-1}) =h(S_1^{N}|W_2),
\end{split}
\end{align*}
where ($a$) follows from the fact that $Y_{2i}$ is a function of $(X_{2i},S_{1i})$ and $X_{2i}$ is a function of $(W_2, Y_2^{i-1})$;
($b$) follows from the fact that $X_2^{i}$ is a function of $(W_2,Y_2^{i-1})$ and $S_1^{i-1}$ is a function of $(Y_2^{i-1},X_2^{i-1})$; ($c$) follows from the fact that $Y_{2}^{i-1}$ is a function of $(X_2^{i-1},S_1^{i-1})$ and $X_2^{i}$ is a function of  $(W_2,S_1^{i-1})$ (by Claim \ref{claim-3}).
\end{proof}

\begin{claim}
\label{claim-3}
For all $i\geq 1$, $X_1^{i}$ is a function of $(W_1,S_2^{i-1})$ and $X_2^{i}$ is a function of $(W_2,S_1^{i-1})$.
\end{claim}
\begin{proof}
By symmetry, it is enough to prove only one. Notice that $X_2^{i}$ is a function of ($W_1,W_2,Z_{1}^{i-2},Z_{2}^{i-1})$. In Fig. \ref{fig:GaussianDIC}, we can easily see that the information of $(W_1,Z_{1}^{i-2})$ delivered to the second link must pass through $S_{1}^{i-1}$. Also $S_{1}^{i-1}$ contains $Z_{2}^{i-1}$. Therefore, $X_2^{i}$ is a function of $(W_2,S_1^{i-1})$.
\end{proof}

Now go back to the main stream of the proof. If $(R_1,R_2)$ is achievable, then $\epsilon_N \rightarrow 0$ as $N \rightarrow \infty$. Hence, we get
\begin{align*}
R_1 + R_2 &\leq  h(Y_{1}|S_{1},X_{2}) - h(Z_1) + h(Y_{2}) - h(Z_2).
\end{align*}
Assume that $X_1$ and $X_2$ have covariance $\rho$, i.e., $E[X_1X_2^*]=\rho$. Then, we get
\begin{align}
\begin{split}
\label{eq:entropyofY2}
h(Y_2) 
\leq \log 2 \pi e \left( 1 + \mathsf{SNR} + \mathsf{INR} + 2 |\rho| \sqrt{\mathsf{SNR} \cdot \mathsf{INR}} \right)
\end{split}
\end{align}
Given $(X_2,S_1)$, the variance of $Y_1$ is upper bounded by
\begin{align*}
\begin{split}
\textrm{Var} \left[ Y_1| X_2, S_1 \right] &\leq K_{Y_1} - K_{Y_1 (X_2, S_1)}K_{(X_2,S_1)}^{-1}K_{Y_1 (X_2, S_1)}^{*},
\end{split}
\end{align*}
where
\begin{align}
\begin{split}
K_{Y_1} &= E\left[ |Y_1|^2 \right] = 1  + \mathsf{SNR} + \mathsf{INR} +  \rho g_d^* g_c + \rho^* g_dg_c^* \\
K_{Y_1 (X_2,S_1)} &= E \left[ Y_1 [X_2^*, S_1^*]  \right] = \left[\rho g_d + g_c, g_c^* g_d + \rho^* \mathsf{INR} \right] \\
K_{(X_2,S_1)} & = E \left[ \left[
                             \begin{array}{cc}
                               |X_2|^2 & X_2 S_1^* \\
                               X_2^* S_1 & |S_1|^2 \\
                             \end{array}
                           \right]
 \right] = \left[
             \begin{array}{cc}
               1 & \rho^* g_c^* \\
               \rho g_c & 1+ \mathsf{INR} \\
             \end{array}
           \right].
\end{split}
\end{align}
By further calculation, 
\begin{align}
\begin{split}
\label{eq:conditonalentropy}
h(Y_1|X_2,S_1) 
\leq \log 2 \pi e \left( 1+  \frac{(1-|\rho|^2)\mathsf{SNR}}{1 + (1-|\rho|^2) \mathsf{INR}} \right)
\end{split}
\end{align}
From (\ref{eq:entropyofY2}) and (\ref{eq:conditonalentropy}),  we get the desired upper bound.

\end{proof}

\subsection{Symmetric Capacity to Within 1.7075 Bits}

\begin{theorem}
\label{theorem:SymmetricCapacityGap}
For all channel parameters $\mathsf{SNR}$ and $\mathsf{INR}$, we can achieve all rates $R$ up to $\overline{C}_{\mathsf{sym}}  - 1.7075$. Therefore, the feedback symmetric capacity $C_{\mathsf{sym}}$ satisfies
\begin{align}
\overline{C}_{\mathsf{sym}} - 1.7075 \leq C_{\mathsf{sym}} \leq \overline{C}_{\mathsf{sym}}.
\end{align}
\end{theorem}
\begin{proof}

In the weak interference regime, we get

\begin{align*}
&2(\overline{C}_{\mathsf{sym}} - R_{\mathsf{sym}}^{\mathsf{weak}}) \leq
  \log \left(1 + \frac{\mathsf{SNR}}{ 1 + \mathsf{INR}} \right) -  \log \left(1 + \frac{\mathsf{SNR}}{ 2 \mathsf{INR}} \right) \\
& + \log \left( 1 + \mathsf{SNR} + \mathsf{INR} + 2 \sqrt{\mathsf{SNR} \mathsf{INR}} \right)  - \log \left(1 + \frac{\mathsf{SNR}}{ 2 \mathsf{INR}} \right) \\
& - \log \left( 1 + \mathsf{INR} \right) + 1 \\
& \leq \log \left(  \frac{1 + \mathsf{SNR} + \mathsf{INR} }{1 + \mathsf{INR}}  \frac{2 \mathsf{INR} }{2 \mathsf{INR} +\mathsf{SNR} }  \right) +  \\
& \quad \log \left(  \frac{1 + \mathsf{SNR} + \mathsf{INR} + 2 \sqrt{\mathsf{SNR} \mathsf{INR}} }{ \mathsf{SNR} + 2\mathsf{INR}  }   \frac{2 \mathsf{INR}  }{ 1 + \mathsf{INR} } \right) + 1 \\
& \leq \log 2 + \log \left( \frac{8}{3} \right) + 1 \approx 3.4150
\end{align*}
Therefore, $\overline{C}_{\mathsf{sym}} - R_{\mathsf{sym}}^{\mathsf{weak}} \leq 1.7075$.

In the strong interference regime, we get
\begin{align*}
&2(\overline{C}_{\mathsf{sym}} - R_{\mathsf{sym}}^{\mathsf{strong}}) \leq
  \log \left(1 + \frac{\mathsf{SNR}}{ 1 + \mathsf{INR}} \right)   \\
 & + \log \left(  \frac{1 + \mathsf{SNR} + \mathsf{INR} + 2 \sqrt{\mathsf{SNR} \mathsf{INR}} }{ 1 + \mathsf{INR} } \right)    \\
& \leq \log 2 + \log 4 = 3
\end{align*}
Therefore, $\overline{C}_{\mathsf{sym}} - R_{\mathsf{sym}}^{\mathsf{strong}} \leq 1.5$.
This completes the proof.
\end{proof}

\subsection{Comparison to Related Work \cite{Kramer:it02, Kramer:it04, GastparKramer:06}}
\label{sec:Comparison}

For the Gaussian channel, Kramer developed a feedback strategy based on Schalkwijk-Kailith scheme and Ozarow's scheme. However, since the scheme is not expressed as a closed form, we cannot see how his scheme is close to the approximate symmetric capacity we derived. To see this, we find the generalized degrees-of-freedom of his scheme.

\begin{lemma}
The generalized degrees-of-freedom of Kramer's scheme is given by
\begin{align}
\underline{d}(\alpha) = \left\{
                                \begin{array}{ll}
                                  1- \alpha, & \hbox{$0 \leq \alpha < \frac{1}{3}$;} \\
                                  \frac{3-\alpha}{4}, & \hbox{$\frac{1}{3} \leq \alpha < 1$;} \\
                                  \frac{1+\alpha}{4}, & \hbox{$\alpha \geq 1$.}
                                \end{array}
                              \right.
\end{align}
\end{lemma}
\begin{proof}
Let $\mathsf{INR} = \mathsf{SNR}^{\alpha}$. Then, by (29) in \cite{Kramer:it02} and (77*) in \cite{Kramer:it04}, we get
\begin{align}
R_{\mathsf{sym}} = \log \left( \frac{1 + \mathsf{SNR} + \mathsf{SNR}^{\alpha} + 2\rho^{*} \mathsf{SNR}^{\frac{\alpha+1}{2}} }{ 1 + (1-\rho^{*2}) \mathsf{SNR}^{\alpha}}  \right),
\end{align}
where $\rho^{*}$ is the solution between 0 and 1 such that
\begin{align*}
&2 \mathsf{SNR}^{\frac{3 \alpha +1 }{2}} \rho^{*4} +  \mathsf{SNR}^{ \alpha } \rho^{*3} - 4 ( \mathsf{SNR}^{\frac{3 \alpha +1 }{2}} + \mathsf{SNR}^{\frac{\alpha +1 }{2}} ) \rho^{*2} \\
&- ( 2 + \mathsf{SNR} + 2 \mathsf{SNR}^{ \alpha } ) \rho^{*} + 2 ( \mathsf{SNR}^{\frac{3 \alpha +1 }{2}} + \mathsf{SNR}^{\frac{\alpha +1 }{2}} ) = 0.
\end{align*}
Notice that for $0 \leq \alpha \leq \frac{1}{3}$, $\mathsf{SNR}$ is a dominant term for high $\mathsf{SNR}$ and $\rho^{*}$ is between 0 and 1; hence, we get $\rho^{*} \approx 2 \mathsf{SNR}^{\frac{3 \alpha -1 }{2}}$. From this, we get $\lim_{\mathsf{SNR} \rightarrow \infty} \frac{R_{\mathsf{sym}}}{\log(\mathsf{SNR})}  = 1 - \alpha$. For $\frac{1}{3} <  \alpha < 1$, the first and second dominant terms are $\mathsf{SNR}^{\frac{3 \alpha +1}{2}}$ and $\mathsf{SNR}$. Also for this range, $\rho^{*}$ is very close to 1. Hence, we approximately get $1-\rho^{*2} \approx \mathsf{SNR}^{\frac{-3 \alpha+1}{4}}$. Therefore, we get the desired result for this range.
For $\alpha \geq 1$, note that the first and second dominant terms are $\mathsf{SNR}^{\frac{3 \alpha +1}{2}}$ and $\mathsf{SNR}$; and $\rho^{*}$ is very close to 1. So we get $1-\rho^{*2} \approx \mathsf{SNR}^{-\frac{ \alpha+1}{4}}$. From this, we get the desired result.
\end{proof}

\begin{figure}[h]
\begin{center}
{\epsfig{figure=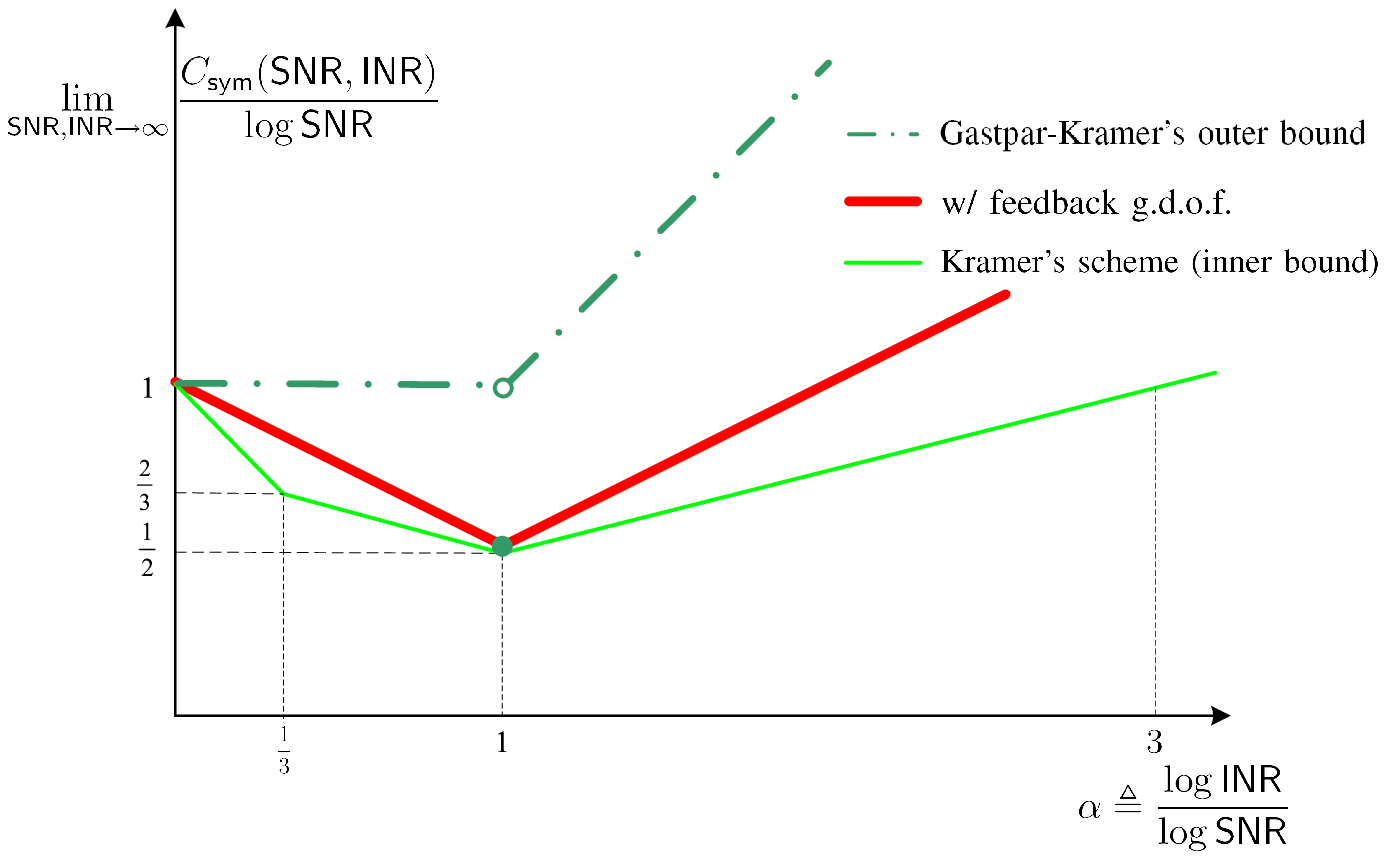, angle=0, width=0.7\textwidth}}
\end{center}
\caption{The generalized degrees-of-freedom of the optimum and Kramer's scheme} \label{fig:gdof-compare}
\end{figure}

Note that in Fig. \ref{fig:gdof-compare} his scheme has the unbounded gap with capacity for all values $\alpha$ except $\alpha = 1$. To compare the scheme to our results for $\alpha = 1$, we also plot the symmetric rate for finite channel parameters as shown in Fig. \ref{fig:ComparisonFinite}. Notice that Kramer's scheme is very close to the outer bound when $\mathsf{INR}$ is similar to $\mathsf{SNR}$. In fact, we can see capacity theorem in \cite{Kramer:it04}, i.e., the Kramer's scheme achieves the capacity when $\mathsf{INR} = \mathsf{SNR} - \sqrt{2 \mathsf{SNR}}$. However, if $\mathsf{INR}$ is quite different from $\mathsf{SNR}$, it becomes far away from the outer bound. Also note that our new bound is much better than Gastpar-Kramer's outer bounds \cite{Kramer:it02, GastparKramer:06}.

\begin{figure}[h]
\begin{center}
{\epsfig{figure=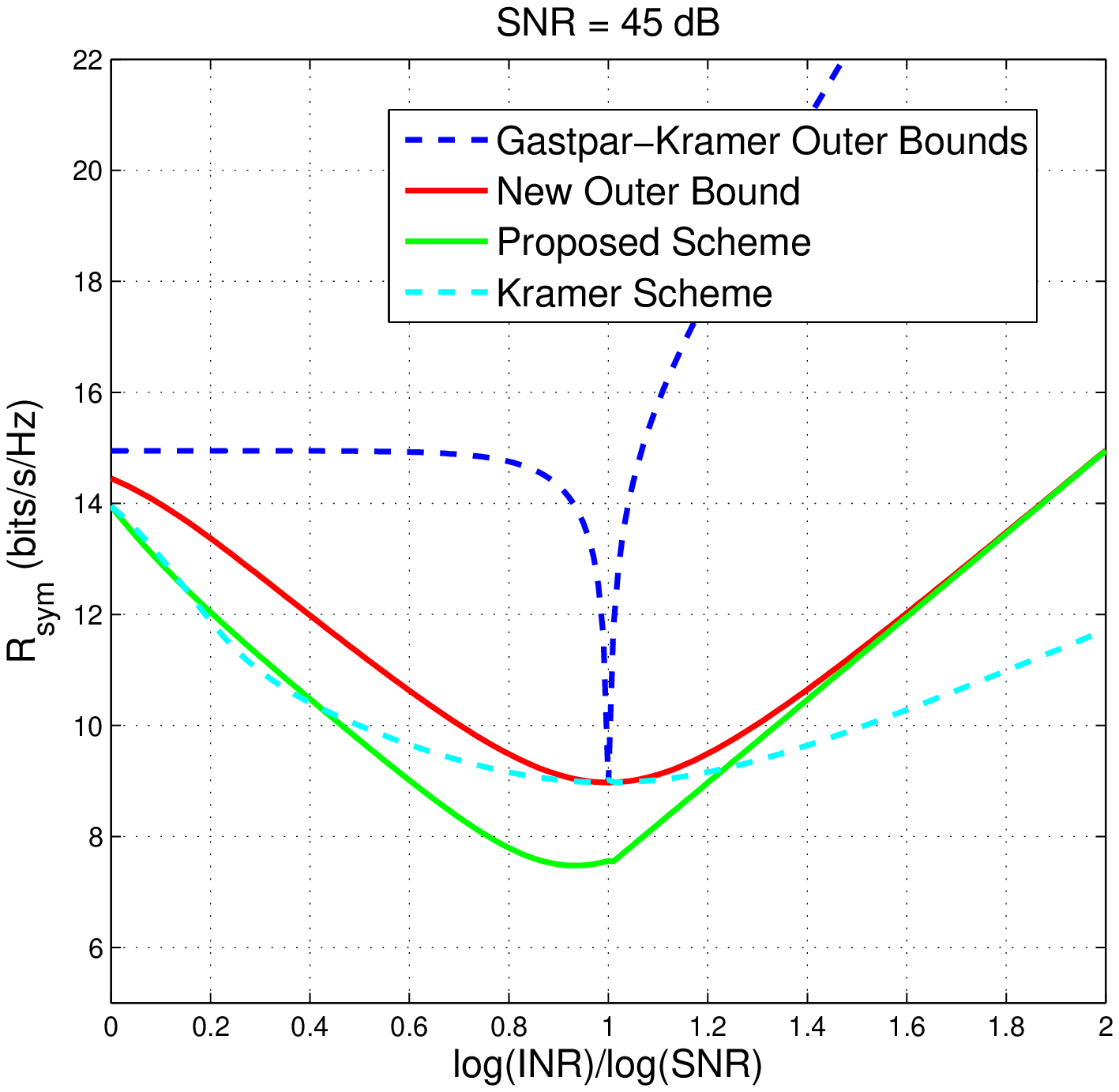, angle=0, width=0.7\textwidth}}
\end{center}
\caption{The symmetric rate of the proposed scheme and Kramer's scheme} \label{fig:ComparisonFinite}
\end{figure}

\section{Discussion and Conclusion}
\label{sec-Conclusion}

Gaining insights from a deterministic model, we found the symmetric capacity to within 1.7075 bits/s/Hz for the two-user Gaussian interference channel with \emph{feedback}. The achievable scheme has two stages and employs the simplified Han-Kobayashi scheme. The second stage is to help refining the corrupted signal received in the first stage. To make desired signals \emph{orthogonal} to interference, we adopted Alamouti's scheme. The constant-gap result is due to a new upper bound.


From this result, we discovered a significant role of feedback that it could provide \emph{unbounded} gain in \emph{many-to-many} channels.
As shown in Fig. \ref{fig:gdof}, we can see feedback gain in two regimes. In the strong interference channel ($\alpha>2$), the gain is because the very strong interference link provides a better alternative path (with the help of feedback) to \emph{relay} information of the other user.
This concept coincides with \emph{correlation routing} in \cite{Kramer:it02}.
The concept is intuitive. On the other hand, in the weak interference channel, there is no better alternative path. However, it turns out that unbounded gain of feedback can be also obtained even in this regime. This is quite surprising because it is counterintuitive.

As a side generalization of a linear deterministic model, we can find the symmetric feedback capacity for a class of deterministic interference channels (El Gamal-Costa's model \cite{ElGamal:it82}). We described the detailed result in Appendix \ref{Appendix:ElGamalCostaModelResult}. Finally, as future work, we need to extend our results into the general \emph{asymmetric} channel with \emph{noisy} feedback.


\appendices

\section{El Gamal-Costa's Deterministic Model}
\label{Appendix:ElGamalCostaModelResult}

El Gamal-Costa's deterministic model \cite{ElGamal:it82} is a generalized version of the linear deterministic model. The channel is described only by the following conditions:
\begin{align}
\begin{split}
H(Y_1|X_1) = H(V_2), \\
H(Y_2|X_2) = H(V_1). \\
\end{split}
\end{align}
Even in this general channel, we can characterize the exact symmetric capacity.

\begin{theorem}
\label{theorem:ElGamalCostaModel}
For El Gamal-Costa's deterministic interference channel, the symmetric feedback capacity is given by
\begin{align*}
\begin{split}
C_{\mathsf{sym}} =  \max_{p(u)p(x_1|u)p(x_2|u)} &\min \left\{ I(U;Y_1) + H(Y_1|V_2,U), \right. \\
& H(Y_2|X_2,U) + H(Y_1|V_1,V_2,U), \\
& \frac{1}{2} \left( H(Y_2) + H(Y_1|V_1,V_2,U) \right), \\
& \left. I(U;Y_1) +  H(Y_1|V_1,U) \right\},
\end{split}
\end{align*}
where $|\mathcal{U}| \leq \min (|\mathcal{V}_1||\mathcal{V}_2|, |\mathcal{Y}_1|,|\mathcal{Y}_2|)$.
\end{theorem}

\textbf{Achievability Proof}: We can adapt the result in \cite{Jiang:07} which found an achievable rate region in discrete memoryless interference channels with feedback. Specializing to the symmetric case, the achievable region is given by the union of all  $R_{\mathsf{sym}}$ for which
\begin{align}
\begin{split}
\label{eq:DMICF_achievable}
    R_{\mathsf{sym}} &\leq I(U_0,Y_1) + I(U_1,X_1;Y_1|U_2,U_0) \\
    R_{\mathsf{sym}} &\leq I(U_1;Y_2|X_2,U_2,U_0) + I(X_1;Y_1|U_2,U_1,U_0) \\
    R_{\mathsf{sym}} &\leq \frac{1}{2}\left( I(U_0;Y_2) + I(U_1,U_2,X_2;Y_2|U_0) + I(X_1;Y_1|U_1,U_2,U_0)\right) \\
    R_{\mathsf{sym}} &\leq I(U_0,Y_1) + I(U_2,X_1;Y_1|U_1,U_0)
\end{split}
\end{align}
over all joint distributions
\begin{align}
\begin{split}
\label{eq:complexjpdf}
p(u_0,u_1,u_2,x_1,x_2)=p(u_0)p(u_1|u_0)p(u_2|u_0) p(x_1|u_1,u_0) p(x_2|u_2,u_0),
\end{split}
\end{align}
where $u_0$, $u_1$, and $u_2$ are the realizations of three auxiliary random variables
$U_0$, $U_1$, and $U_2$ defined on arbitrary finite sets $\mathcal{U}_0$, $\mathcal{U}_1$, and $\mathcal{U}_2$.

The idea of the scheme is to combine the block Markov encoding \cite{Cover:it79, Cover:it81} and the Han-Kobayashi scheme \cite{HanKoba:it81}.
Specializing to the deterministic case, we set
\begin{align}
\label{eq:ViUi}
U_1 = V_1, \; U_2 = V_2.
\end{align}
Since the role of $U_0$ is to reflect feedback, we still need this random variable. We replace $U_0$ with $U$ for notational simplicity.

Now let us check a joint distribution. Using functional relationship between $X_i$ and $V_i$, we can simplify (\ref{eq:complexjpdf}) into the form of $p(u,x_1,x_2)$.
The idea is to write $p(u,v_1,v_2,x_1,x_2)$ into two different ways:
\begin{align}
\begin{split}
\label{eq:jointdist}
p(u,v_1,v_2,x_1,x_2) &= p(u)p(x_1|u)p(x_2|u) \delta (v_1 - g_1(x_1)) \delta (v_2-g_2(x_2) \\
&= p(u)p(v_1|u)p(v_2|u) p(x_1|v_1,u) p(x_2|v_2,u)
\end{split}
\end{align}
where $\delta(\cdot)$ indicates the Kronecker delta function and
\begin{align}
p(x_1|v_1,u) := \frac{p(x_1|u) \delta(v_1 - g_1(x_1))}{p(v_1|u)}, \\
p(x_2|v_2,u) := \frac{p(x_2|u) \delta(v_2 - g_2(x_1))}{p(v_2|u)}.
\end{align}
Therefore, knowing only $p(u,x_1,x_2)$ is enough to generate codebook.
From (\ref{eq:ViUi}) and (\ref{eq:jointdist}), we can get the desired form of the joint distribution. Therefore, we establish the achievability proof.

\textbf{Converse Proof}:
One of the main points is how to introduce an auxiliary random variable $U$. We choose
$U_i = (V_1^{i-1},V_2^{i-1})$. In Claim \ref{claim:converseIndependence}, we will show that given $U_i$, $X_{1i}$ and $X_{2i}$ are conditionally independent.

Consider the first upper bound.

\begin{align*}
\begin{split}
N&R_1 = H(W_1) \overset{(a)} \leq I(W_1;Y_1^{N}) +  N \epsilon_N\\
&\overset{(b)}{\leq} \sum \left[ I(Y_{1i};U_i) + H(Y_{1i}|U_i) \right] - \sum H(Y_{1i}|Y_1^{i-1},W_1,U_i)  + N \epsilon_N \\
&\overset{(c)}{\leq} \sum \left[ I(Y_{1i};U_i) + H(Y_{1i},V_{2i}|U_i) - H(V_{2i}|U_i) \right]  + N \epsilon_N \\
&= \sum \left[ I(Y_{1i};U_i) + H(Y_{1i}|V_{2i},U_i) \right]   + N \epsilon_N,
\end{split}
\end{align*}
where $(a)$ follows from Fano's inequality; $(b)$ follows from the fact that conditioning reduces entropy and $H(Y_1^N|W_1)=\sum H(Y_{1i}|Y_1^{i-1},W_1,U_i)$ (see Claim \ref{claim:converserelation}); and ($c$) follows from the fact that $\sum H(V_{2i}|U_i) = \sum H(Y_{1i}|Y_1^{i-1},W_1,U_i)$ (see Claim \ref{claim:converserelation}) and adding information increases mutual information.

Now consider the second upper bound.

\begin{align*}
\begin{split}
N&R_1 = H(W_1) = H(W_1|W_2) \leq I(W_1;Y_1^N|W_2) + N \epsilon_N \\
&\leq \sum H(Y_{1i}|Y_1^{i-1},W_2) + N \epsilon_N \\
&\overset{(a)}{=} \sum H(Y_{1i}|Y_1^{i-1},W_2,X_{2i},U_i) + N \epsilon_N \\
&\overset{(b)}{\leq} \sum H(Y_{1i},Y_{2i}|Y_1^{i-1},W_2,X_{2i},U_i) + N \epsilon_N \\
&\overset{(c)}{=} \sum H(Y_{2i}|Y_1^{i-1},W_2,X_{2i},U_i) + \sum H(Y_{1i}|Y_1^{i-1},W_2, X_{2i},U_{i},Y_{2i},V_{1i}) + N \epsilon_N \\
&\overset{(d)}{\leq} \sum \left[ H(Y_{2i}|X_{2i}, U_i) + H(Y_{1i}|V_{1i},V_{2i},U_i) \right] + N \epsilon_N
\end{split}
\end{align*}
where ($a$) follows from the fact that $X_{2}^{i}$, $V_2^i$, $X_1^{i-1}$, $V_1^{i-1}$ are functions of $(W_2,Y_1^{i-1})$, $X_2^i$,  $(Y_1^{i-1},V_2^{i-1})$, $X_1^{i-1}$, respectively (see Claim \ref{claim:converseFunctional}); ($b$) follows from the fact that adding information increases entropy; ($c$) follows from the chain rule and $V_{1i}$ is a function of $(Y_{2i},X_{2i})$; ($d$) is because conditioning reduces entropy.

Consider the third one.
\begin{align*}
\begin{split}
N&(R_1 + R_2)= H(W_1) + H(W_2) =  H(W_1|W_2)+ H(W_2) \\
&\leq I(W_1;Y_1^{N}|W_2) + I(W_2;Y_2^{N}) + N \epsilon_N \\
&= H(Y_1^{N}|W_2) + I(W_2;Y_2^{N}) + N \epsilon_N \\
&\leq H(Y_1^{N},V_1^{N}|W_2) + I(W_2;Y_2^{N}) + N \epsilon_N \\
&=H(Y_1^{N}|V_1^{N},W_2) + H(Y_2^{N}) + \left[H(V_1^{N}|W_2) - H(Y_2^{N}|W_2)\right] +  N \epsilon_N \\
&\overset{(a)}{=} H(Y_1^{N}|V_1^{N},W_2) + H(Y_2^{N}) +  N \epsilon_N \\
&\overset{(b)}{=} H(Y_1^{N}|V_1^{N},W_2,X_2^{N},V_2^{N}) + H(Y_2^{N}) +  N \epsilon_N \\
&\overset{(c)}{\leq} \sum \left[ H(Y_{1i}|V_{1i},V_{2i},U_i) + I(U_i;Y_{2i}) + H(Y_{2i}|U_i) \right] + N \epsilon_N
\end{split}
\end{align*}
where ($a$) follows from $H(V_1^N|W_2)=H(Y_2^N|W_2)$ (by Claim \ref{claim-2}); ($b$) follows from the fact that $X_2^{N}$ is a function of $(W_2,V_1^{N-1})$ (by Claim \ref{claim-1}) and $V_2^{N}$ is a function of $X_2^{N}$; ($c$) follows from the fact that conditioning reduces entropy.

Now consider the last part.
\begin{align*}
\begin{split}
n&(R_1 + R_2)= H(W_1) + H(W_2) \leq I(W_1;Y_1^{N}) + I(W_2;Y_2^{N}) + N \epsilon_N \\
&\overset{(a)}{\leq} \sum \left[ I(U_i;Y_{1i}) + H(Y_{1i}|U_i) -H(Y_{1i}|Y_1^{i-1},W_1,U_i)  \right]  \\
& + \sum \left[ I(U_i;Y_{2i}) + H(Y_{2i}|U_i) -H(Y_{2i}|Y_2^{i-1},W_2,U_i)  \right]  + N \epsilon_N \\
&\overset{(b)}{\leq} \sum \left[ I(U_i;Y_{1i}) + H(Y_{1i},V_{1i}|U_i) -H(V_{2i}|U_i)  \right]  + \sum \left[ I(U_i;Y_{2i}) + H(Y_{2i},V_{2i}|U_i) -H(V_{1i}|U_i)  \right]  + N \epsilon_N \\
&= \sum \left[ I(U_i;Y_{1i}) + H(Y_{1i}|V_{1i},U_i) \right]  + \sum \left[ I(U_i;Y_{2i}) + H(Y_{2i}|V_{2i},U_i)  \right]  + N \epsilon_N \\
\end{split}
\end{align*}
where ($a$) follows from the fact that conditioning reduces entropy and $H(Y_2^N|W_2)= \sum H(Y_{2i}|Y_2^{i-1},W_2,U_i)$ by Claim \ref{claim:converserelation}; ($b$) follows from the fact that $\sum H(V_{2i}|U_i) = \sum H(Y_{1i}|Y_{1}^{i-1},W_1,U_i)$, $\sum H(V_{1i}|U_i) = \sum H(Y_{2i}|Y_{2}^{i-1},W_2,U_i)$ (by Claim \ref{claim:converserelation}), and adding information increases entropy.

Let $Q$ be the time index uniformly distributed over the set $\{1,2,\cdots,N \}$ and independent of $(W_1,W_2,X_1^N,X_2^N,Y_1^N,Y_2^N)$.
Define $X_1 = X_{1Q}$, $V_1 = V_{1Q}$, $X_2 = X_{2Q}$, $V_2 = V_{1Q}$, $Y_1 = Y_{1Q}$, $Y_2 = Y_{2Q}$, $U = (U_Q, Q)$. If $(R_1,R_2)$ is achievable, then $\epsilon_N \rightarrow 0$ as $N \rightarrow \infty$. Hence, we obtain
\begin{align*}
R_1 &\leq I(Y_{1};U) + H(Y_{1}|V_{2},U) ],\\
R_1 &\leq H(Y_{2}|X_{2},U) + H(Y_1|V_1,V_2,U),\\
R_1 + R_2 &\leq  I(U;Y_2) + H(Y_{1}|V_{1},V_{2},U) + H(Y_{2}|U),\\
R_1 + R_2 &\leq  I(U;Y_1) + I(U;Y_2) + H(Y_{1}|V_{1}) + H(Y_{2}|V_{2}).
\end{align*}
By Claim \ref{claim:converseIndependence}, $X_1$ and $X_2$ are conditionally independent given $U$. Therefore, $\exists p(u,x_1,x_2)=p(u)p(x_1|u)p(x_2|u)$ such that the desired inequalities hold. This establishes the converse.

\textbf{Several Claims for the Converse Proof}:

\begin{claim}
\label{claim:converseIndependence}
Given $U_i = (V_1^{i-1},V_2^{i-1})$, $X_{1i}$ and $X_{2i}$ are conditionally independent. Consequently,
\begin{align}
H(W_1 | W_2,U)&= H(W_1|U),\\
H(X_1 | X_2,U)&= H(X_1|U),\\
H(V_1 | V_2,U)&= H(V_1|U).
\end{align}
\end{claim}
\begin{proof}
The idea is based on the technique used in \cite{Willems:it82}.
For completeness we describe it thoroughly. For two arbitrary message pairs $(w_1,w_2)$ and $(w_1',w_2')$, we obtain the following relationship:
\begin{align}
\begin{split}
\label{eq-distributionrelation}
&p(u_i|w_1,w_2)p(u_i|w_1',w_2') =p(v_1^{i-1},v_2^{i-1}|w_1,w_2)p(v_1^{i-1},v_2^{i-1}|w_1',w_2') \\
&\overset{(a)}{=}\prod_{j=1}^{i-1} p(v_{1j}|v_1^{j-1},w_1,w_2) p(v_{2j}|v_2^{j-1},v_1^{i-1},w_1,w_2) \cdot p(v_{1j}|v_1^{j-1},w_1',w_2') p(v_{2j}|v_2^{j-1},v_1^{i-1},w_1',w_2') \\
&\overset{(b)}{=} \prod_{j=1}^{i-1} p(v_{1j}|v_2^{j-1},w_1) p(v_{2j}|v_1^{j-1},w_2) \cdot p(v_{1j}|v_2^{j-1},w_1') p(v_{2j}|v_1^{j-1},w_2') \\
&\overset{(c)}{=} \prod_{j=1}^{i-1} p(v_{1j}|v_2^{j-1},w_1') p(v_{2j}|v_1^{j-1},w_2)  \cdot p(v_{1j}|v_2^{j-1},w_1) p(v_{2j}|v_1^{j-1},w_2') \\
&= p(u_i|w_1',w_2)p(u_i|w_1,w_2'),
\end{split}
\end{align}
where ($a$) follows from the chain rule; ($b$) follows from Claim \ref{claim-1}; ($c$) follows from rearranging a product order.

Using this, we obtain
\begin{align*}
\begin{split}
&p(w_1,w_2|u_i) = \frac{p(w_1)p(w_2)p(u_i|w_1,w_2)}{p(u_i)} \\
&= \frac{p(w_1)p(w_2)p(u_i|w_1,w_2)}{p(u_i)} \cdot \frac{\sum_{w_1'}\sum_{w_2'} p(w_1')p(w_2')p(u_i|w_1',w_2')}{p(u_i)} \\
&\overset{(a)}{=} \frac{\sum \sum p(w_1)p(w_2) p(w_1')p(w_2')p(u_i|w_1',w_2) p(u_i|w_1,w_2')}{p(u_i)p(u_i)} \\
& = \frac{\sum_{w_2'} p(w_1) p(w_2')p(u_i|w_1,w_2')}{p(u_i)} \cdot \frac{\sum_{w_1'} p(w_2)p(w_1')p(u_i|w_1',w_2)}{p(u_i)} \\
& = p(w_1|u_i) \cdot p(w_2|u_i),
\end{split}
\end{align*}
where $(a)$ follows from (\ref{eq-distributionrelation}). This proves the independence of $W_1$ and $W_2$ given $u_i$.

Also it follows easily that
\begin{align}
\begin{split}
&p(x_{2i}|u_i,x_{1i}) \overset{(a)}{=} p(f^{i}(W_2,v_1^{i-1})|v_1^{i-1},v_2^{i-1},f^{i}(W_1,v_2^{i-1})) \\
&\overset{(a)}{=}p(f^{i}(W_2,v_1^{i-1})|v_1^{i-1},v_2^{i-1}) =p(x_{2i}|u_i),
\end{split}
\end{align}
where $(a)$ follows from Claim \ref{claim-1} and $(b)$ follows from the independence of $W_1$ and $W_2$ given $u_i$. This
implies that $x_{1i}$ and $x_{2i}$ are independent given $u_i$. Since $v_{1i},v_{2i}$ are functions of $x_{1i},x_{2i}$, respectively, $v_{1i}$ and $v_{2i}$ are also independent given $u_i$.
Therefore, we complete the proof.
\end{proof}

\begin{claim}
\label{claim:converserelation}
$H(Y_2^{N}|W_2) = \sum H(Y_{2i}|Y_{2}^{i-1},W_2,U_i) = \sum H(V_{1i}|U_i).$
\end{claim}
\begin{proof}

We prove the first equality.
\begin{align*}
\begin{split}
H(Y_2^{N}|W_2) &= \sum H(Y_{2i}|Y_2^{i-1},W_2) \overset{(a)}{=} \sum H(Y_{2i}|Y_2^{i-1},W_2,X_2^{i},V_2^{i},V_1^{i-1}) \\
&=  \sum H(Y_{2i}|Y_2^{i-1},W_2,X_2^{i},U_{i}) \overset{(b)}{=}  \sum H(Y_{2i}|Y_2^{i-1},W_2,U_{i}),
\end{split}
\end{align*}
where $(a)$ follows from the fact that $X_2^{i}$, $V_2^{i}$, $V_1^{i-1}$ are functions of $(W_2,Y_2^{i-1})$, $X_2^{i}$, $(X_2^{i-1},Y_2^{i-1})$, respectively; $(b)$ follows from the fact that $X_2^{i}$ is a function of $(W_2,Y_2^{i-1})$.

Next we prove the second one.
\begin{align*}
\begin{split}
\sum H(V_{1i}|U_i) &\overset{(a)}{=} \sum H(V_{1i}|U_i, W_2) \overset{(b)}{=}  \sum H(V_{1i}|V_1^{i-1},W_2) =  H(V_{1}^N|W_2) \overset{(c)}{=} H(Y_2^{N}|W_2),
\end{split}
\end{align*}
where $(a)$ is because $V_{1i}$ and $W_2$ are conditionally independent given $U_i$ (by Claim \ref{claim:converseIndependence}); and $(b)$ follows from the fact that $V_2^{i-1}$ is a function of $(W_2,V_1^{i-1})$; and $(c)$ follows from Claim \ref{claim-2}.
\end{proof}

\begin{claim}
\label{claim:converseFunctional}
For $i\geq 1$, $X_1^{i}$ is a function of $(W_1,Y_2^{i-1})$ and
$X_2^{i}$ is a function of $(W_2,Y_1^{i-1})$.
\end{claim}
\begin{proof}
By symmetry, it is enough to prove only one. We know from Claim \ref{claim-1} that $X_1^{i}$ is a function of $(W_1,V_2^{i-1})$.
Note that $V_2^{i-1}$ is a function of $X_{2}^{i-1}$ (a function of ($Y_2^{i-1},V_1^{i-1}$)). Also note that $V_1^{i-1}$ is a function of $X_{1}^{i-1}$ (a function of ($W_1, V_2^{i-2}$) by Claim \ref{claim-1}). Hence we know that
\begin{align*}
\begin{split}
X_1^{i} \textrm{ is a function of } (W_1,V_2^{i-2},Y_2^{i-1}).
\end{split}
\end{align*}
Iterating this procedure $(i-3)$ times, we know that $X_1^{i}$ is a function of $(W_1,V_{21},Y_2^{i-1})$. Note that $V_{21}$ is a function of $X_{21}$ (a function of ($V_{11},Y_{21})$) and $V_{11}$ is a function of $X_{11}$. Since $X_{11}$ depends only on $W_1$ due to no feedback in the initial time, we conclude that
\begin{align*}
X_1^{i} \textrm{ is a function of } (W_1,Y_2^{i-1}).
\end{align*}
\end{proof}

\bibliographystyle{ieeetr}

\end{document}